\documentclass{llncs}
\usepackage[latin1,utf8]{inputenc}
\usepackage{scalefnt}
\usepackage{xspace}
\usepackage{amsmath}
\usepackage{graphicx,url}
\usepackage{booktabs}
\usepackage{float}
\usepackage{subfigure}
\usepackage[ruled,linesnumbered,vlined]{algorithm2e}
\usepackage{algpseudocode}
\usepackage{amssymb}
\usepackage{hf-tikz}
\usepackage{marvosym}
\usepackage{textgreek}

\definecolor{RED}{RGB}{255,0,0}

\newtheorem{Claim}{Claim}
\begin{document}

\title{On the Minimum Cycle Cover problem on graphs with bounded co-degeneracy\thanks{This research has received funding from Rio de Janeiro Research Support Foundation (FAPERJ) under grant agreement E-26/201.344/2021,  National Council for Scientific  and Technological Development (CNPq) under grant agreement 309832/2020-9, and the European Research Council (ERC) under the 
\begin{minipage}{0.75\textwidth}
European Union's Horizon $2020$ research and innovation programme under grant agreement CUTACOMBS (No. $714704$). 
\end{minipage}
\begin{minipage}{0.2\textwidth}
%    \begin{center}
        \includegraphics[scale=.63]{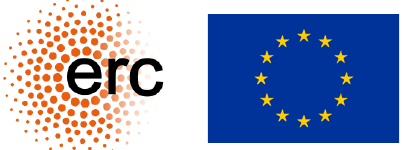}
%    \end{center}
\end{minipage}}}
\titlerunning{On the Minimum Cycle Cover problem on graphs with bounded co-degeneracy}
% If the paper title is too long for the running head, you can set
% an abbreviated paper title here
%
\author{Gabriel L. Duarte\inst{2}%\orcidID{0000-0002-5811-210}
\and
Uéverton S. Souza\inst{1,2}%\orcidID{0000-0002-5007-4513}
}
\authorrunning{G. L. Duarte and U. S. Souza}
\institute{Institute of Informatics, University of Warsaw, Warsaw, Poland 
\and
Instituto de Computação, Universidade Federal Fluminense, Niterói, Brazil
\\
\email{gabrield@id.uff.br},
\email{ueverton@ic.uff.br}
}
\maketitle  

\begin{abstract}
In 2017, Knop, Kouteck{\'y}, Masařík, and Toufar [WG 2017] asked about the complexity of deciding graph problems $\Pi$ on the complement of $G$ considering a parameter $p$ of $G$, especially for sparse graph parameters such as treewidth. In 2021, Duarte, Oliveira, and Souza [MFCS~2021] showed some problems that are FPT when parameterized by the treewidth of the complement graph (called co-treewidth). Since the degeneracy of a graph is at most its treewidth, they also introduced the study of co-degeneracy (the degeneracy of the complement graph) as a parameter. 
In 1976, Bondy and Chv\'{a}tal [DM 1976] introduced the notion of {\em closure} of a graph: let $\ell$ be an integer; the $(n+\ell)$-closure, $\operatorname{cl}_{n+\ell}(G)$, of a graph $G$ with $n$ vertices is obtained from $G$ by recursively adding an edge between pairs of nonadjacent vertices whose degree sum is at least $n+\ell$ until no such pair remains. 
A graph property $\Upsilon$ defined on all graphs of order $n$ is said to be $(n+\ell)$-stable if for any graph $G$ of order $n$ that does not satisfy $\Upsilon$, the fact that $uv$ is not an edge of $G$ and that $G+uv$ satisfies $\Upsilon$ implies $d(u)+d(v)< n+\ell$. 
Duarte et al. [MFCS 2021] developed an algorithmic framework for co-degeneracy parameterization based on the notion of closures for solving problems that are $(n+\ell)$-stable for some $\ell$ bounded by a function of the co-degeneracy.  
In 2019, Jansen, Kozma, and Nederlof [WG 2019] relax the conditions of Dirac’s theorem and consider input graphs $G$ in which at least $n-k$ vertices have degree at least $\frac{n}{2}$, and present an FPT algorithm concerning to $k$, to decide whether such graphs $G$ are Hamiltonian.
In this paper, we first determine the stability of the property of having a bounded cycle cover. After that, combining the framework of Duarte et al. [MFCS 2021] with some results of Jansen et al. [WG 2019], we obtain a $2^{\mathcal{O}(k)}\cdot n^{\mathcal{O}(1)}$-time algorithm for {\sc Minimum Cycle Cover} on graphs with co-degeneracy at most $k$, which generalizes Duarte et al. [MFCS 2021] and Jansen~et~al. [WG 2019] results concerning the {\sc Hamiltonian Cycle} problem.

\end{abstract}
\keywords{degeneracy. complement graph. cycle cover. FPT. kernel}

\section{Introduction}

Graph width parameters are useful tools for identifying tractable classes of instances for NP-hard problems and designing efficient algorithms for such problems on these instances. 
Treewidth and clique-width are two of the most po\-pu\-lar graph width pa\-ra\-me\-ters. 
An algorithmic meta-theorem due to Courcelle, Makowsky, and Rotics~\cite{Co00} states that any
problem expressible in the monadic second-order logic on graphs (MSO$_1$) can be solved in FPT time when parameterized by the clique-width of the input graph.\footnote{Originally this required a clique-width expression as part of the input.} %\footnote{Originally this required a clique-width expression as part of the input. This restriction was removed when Oum and Seymour~\cite{OUM2006514} gave an algorithm that finds a $2^{\mathcal{O}(cw)}$-approximation of an optimal clique-width expression.} 
In addition, Courcelle~\cite{C90} states that any problem expressible in the monadic second-order logic of graphs with edge set quantifications (MSO$_2$) can be solved in
FPT time when parameterized by the treewidth of the input graph. 
Although the class of graphs with bounded treewidth is a subclass of the class of graphs with bounded clique-width~\cite{corneil2005relationship}, the MSO$_2$ logic on graphs extends the MSO$_1$ logic, and there are MSO$_2$
properties like ``$G$ has a Hamiltonian cycle'' that are not
MSO$_1$ expressible~\cite{courcelle1994monadic}. 
In addition, there are problems that are fixed-parameter tractable when parameterized by treewidth, such as {\sc MaxCut}, {\sc Largest Bond}, {\sc Longest Cycle}, {\sc Longest Path}, {\sc Edge Dominating Set}, {\sc Graph Coloring}, {\sc Clique Cover}, {\sc Minimum Path Cover}, and {\sc Minimum Cycle Cover} that cannot be FPT when parameterized by clique-width~\cite{Bond,fomin2009clique,fomin2010algorithmic,fomin2010intractability,fomin2014almost}, unless FPT = W[1].

For problems that are fixed-parameter tractable concerning treewidth, but intractable when parameterized by clique-width, the identification of tractable classes of instances of bounded clique-width and unbounded treewidth becomes a fundamental quest~\cite{MFCS_2021}. In 2016,  Dvo\v{r}{\'{a}}k, Knop, and Masařík~\cite{DBLP:journals/corr/DvorakKM16a} showed that
{\sc $k$-Path Cover} is FPT when parameterized by the treewidth of the complement of the input graph. This implies that {\sc Hamiltonian Path} is FPT when parameterized by the treewidth of the complement graph. In 2017, Knop, Kouteck{\'y}, Masařík, and Toufar (WG 2017,~\cite{knop2017simplified}) asked about the complexity of
deciding graph problems $\Pi$ on the complement of $G$ considering a parameter $p$ of $G$  (i.e., with respect to $p(G)$), especially for sparse graph parameters such as treewidth. In fact, the treewidth of the complement of the input graph, proposed be called {\em co-treewidth} in~\cite{MFCS_2021}, seems a nice width parameter to deal with dense instances of problems that are hard concerning clique-width. {\sc MaxCut},
{\sc Clique Cover}, and {\sc Graph Coloring} are example of problems W[1]-hard concerning clique-width but FPT-time solvable when parameterized by co-treewidth (see~\cite{MFCS_2021}).

The \emph{degeneracy} of a graph $G$ is the least $k$ such that every induced subgraph of $G$ contains a vertex with degree at most $k$. {Equivalently, the degeneracy of $G$ is the least $k$ such that its vertices can be arranged into a sequence so that each vertex is adjacent to at most $k$ vertices preceding it in the sequence.} It is well-known that the degeneracy of a graph is upper bounded by its treewidth; thus, the class of graphs with bounded treewidth is also a subclass of the class of graphs with bounded degeneracy. In~\cite{MFCS_2021}, Duarte, Oliveira, and Souza presented an algorithmic framework to deal with the degeneracy of the complement graph, called {\em co-degeneracy}, as a parameter. %They showed that {\sc Longest Path}, {\sc Longest Cycle}, {\sc Path Cover}, and {\sc Edge Dominating Set} are all fixed-parameter tractable when parameterized by the co-degeneracy of the graph.% (thus, they are also FPT concerning co-treewidth).  

Although the notion of co-parameters is as natural as their complementary versions, just a few studies have ventured into the world of dense instances with respect to sparse parameters of their complements. %, perhaps because they are hidden behind the complement graph.
Also, note that would be natural to consider \emph{``co-clique-width''} parameterization, but Courcelle and Olariu~\cite{CO00-2} proved that for every graph $G$ its clique-width is at most twice the clique-width of $\overline{G}$. Thus, the co-clique-width notion is redundant from the point of view of parameterized complexity. 
%This fact indicates that for general purposes, the most interesting co-parameters tend to be incomparable with clique-width. 
Therefore, in the sense of being a useful parameter for many NP-hard problems in identifying a large and new class of (dense) instances that can be efficiently handled, the co-degeneracy seems interesting because it is incomparable with clique-width and stronger\footnote{A parameter $y$ is stronger than $x$, if the set of instances where $x$ is bounded is a subset of those where $y$ is bounded.} than co-treewidth. 

In~\cite{MFCS_2021}, Duarte, Oliveira, and Souza developed an algorithmic framework for co-degeneracy parameterization based on the notion of Bondy-Chv\'{a}tal closure for solving problems that have a ``bounded'' stability concerning some closure. More precisely, for a graph $G$ with $n$ vertices, and two distinct nonadjacent vertices $u$ and $v$ of $G$ such that $d(u)+d(v)\geq n$, Ore's theorem states that $G$ is hamiltonian if and only if $G+uv$ is hamiltonian. In 1976, Bondy and Chv\'{a}tal~\cite{bondy1976method} generalized Ore's theorem and defined the {\em closure} of a graph: 
\vspace{-.1cm}
\begin{itemize}
    \item let $\ell$ be an integer; the $(n+\ell)$-closure, $\operatorname{cl}_{n+\ell}(G)$, of a graph $G$ is obtained from $G$ by recursively adding an edge between pairs of nonadjacent vertices whose degree sum is at least $n+\ell$ until no such pair remains.
\vspace{-.1cm}
\end{itemize}
 Bondy and Chv\'{a}tal showed that $\operatorname{cl}_{n+\ell}(G)$ is uniquely determined from $G$ and that $G$ is hamiltonian if and only if $\operatorname{cl}_n(G)$ is hamiltonian.
 
A property $\Upsilon$ defined on all graphs of order $n$ is said to be $(n+\ell)$-stable if for any graph $G$ of order $n$ that does not satisfy $\Upsilon$, the fact that $uv$ is not an edge of $G$ and that $G+uv$ satisfies $\Upsilon$ implies $d(u)+d(v)< n+\ell$.  In other words, if $uv\notin E(G)$, $d(u)+d(v)\geq n+\ell$ and $G+uv$ has property $\Upsilon$, then $G$ itself has property $\Upsilon$ (c.f.~\cite{broersma2000closure}). The smallest integer $n+\ell$ such that $\Upsilon$ is $(n+\ell)$-stable is the \emph{stability} of $\Upsilon$, denoted by $s(\Upsilon)$. Note that Bondy and Chv\'{a}tal showed that Hamiltonicity is $n$-stable. A survey on the stability of graph properties can be found in~\cite{broersma2000closure}.

In~\cite{MFCS_2021}, based on the fact that the class of graphs with co-degeneracy at most $k$ is closed under completion (edge addition), it was proposed the following framework for determining whether a graph $G$ satisfies a property $\Upsilon$ in FPT time regarding the co-degeneracy of $G$, denoted by $k$: 

%\newpage

\begin{enumerate}
    \item determine an upper bound for $s(\Upsilon)$ - the stability of $\Upsilon$; 
    \item If $s(\Upsilon)\leq n+\ell$ where $\ell\leq f(k)$ (for some computable function $f$) then   
    \begin{enumerate}
        \item set $G = \operatorname{cl}_{n+\ell}(G)$;
        \item since $G=\operatorname{cl}_{n+\ell}(G)$ and $G$ has co-degeneracy $k$ then $G$ has co-vertex cover number (distance to clique) at most $2k+\ell+1$ (see~\cite{MFCS_2021});
        \item at this point, it is enough to solve the problem in FPT-time concerning co-vertex cover parameterization.\vspace{-.1cm}  
    \end{enumerate}
\end{enumerate}

%On the other hand, the notion of co-treewidth seems to be interesting given that bounded co-treewidth implies bounded clique-width, and treewidth and co-treewidth are incomparable parameters. Moreover, although co-degeneracy is incomparable with clique-width, it is stronger than co-treewidth and a useful parameter for handling some problems on dense instances. In this paper,  we show that {\sc Longest Path},  {\sc Longest Cycle},  and {\sc Edge Dominating Set} are FPT when parameterized by co-degeneracy,  while {\sc Graph Coloring} is para-NP-hard. The complexity of {\sc MaxCut} parameterized by co-degeneracy is left open.

In~\cite{MFCS_2021}, using such a framework, it was shown that {\sc Hamiltonian Path}, {\sc Hamiltonian Cycle}, {\sc Longest Path},  {\sc Longest Cycle}, and {\sc Minimum Path Cover} are all fixed-parameter tractable when parameterized by co-degeneracy.
Note that {\sc Longest Path} and {\sc Minimum Path Cover} are two distinct ways to generalize the {\sc Hamiltonian Path} problem just as {\sc Longest Cycle} and {\sc Minimum Cycle Cover} generalize the {\sc Hamiltonian Cycle} problem. However, the {\sc Minimum Cycle Cover} problem seems to be more challenging than the others concerning co-degeneracy parameterization, even because the stability of having a cycle cover of size at most $r$, to the best of our knowledge, is unknown. 

In the {\sc Minimum Cycle Cover} problem, we are given a simple graph $G$ and asked to find a minimum set $S$ of vertex-disjoint cycles of $G$ such that each vertex of $G$ is contained in one cycle of $S$, where single vertices are considered trivial cycles. Note that each nontrivial cycle has size at least three. In this paper, our focus is on {\sc Minimum Cycle Cover} parameterized by co-degeneracy.

The Dirac’s theorem from 1952 (see~\cite{dirac1952some}) states that a graph $G$ with $n$ vertices ($n\geq 3$) is Hamiltonian if every vertex of $G$ has degree at least $\frac{n}{2}$. In~\cite{jansen2019hamiltonicity}, Jansen, Kozma, and Nederlof relax the conditions of Dirac’s theorem and consider input graphs $G$ in
which at least $n-k$ vertices have degree at least $\frac{n}{2}$, and present a $2^{\mathcal{O}(k)}\cdot n^{\mathcal{O}(1)}$-time algorithm to decide whether $G$ has a Hamiltonian cycle. 
In~2022, F. Fomin, P. Golovach, D. Sagunov, and K. Simonov~\cite{SODA_Fedor} presented the following algorithmic generalization of Dirac's theorem: if all but $k$ vertices of a 2-connected graph $G$ are of degree at least $\delta$, then deciding whether $G$ has a cycle of length at least $\min\{2\delta + k, n\}$ can be done in time $2^k \cdot n^{\mathcal{O}(1)}$.
Besides, in 2020, F. Fomin, P. Golovach, D. Lokshtanov, F. Panolan, S. Saurabh, and M. Zehavi~\cite{SJDM_fahad} proved that deciding whether a $2$-connected $d$-degenerate $n$-vertex $G$ contains a cycle of length at least $d+k$ can be done in time $2^{\mathcal{O}(k)}\cdot n^{\mathcal{O}(1)}$.

In this paper, we first determine the stability of the property of having a cycle cover of size at most $r$. After that,  using the closure framework proposed in~\cite{MFCS_2021} together with some results and techniques presented in~\cite{jansen2019hamiltonicity}, we show that {\sc Minimum Cycle Cover} admits a kernel with linear number of vertices when parameterized by co-degeneracy. After that, by designing an exact single-exponential time algorithm for solving {\sc Minimum Cycle Cover}, we obtain as a corollary a $2^{\mathcal{O}(k)}\cdot n^{\mathcal{O}(1)}$-time algorithm for the {\sc Minimum Cycle Cover} problem on graphs with co-degeneracy at most $k$. These results also implies a $2^{\mathcal{O}(k)}\cdot n^{\mathcal{O}(1)}$-time algorithm for solving {\sc Minimum Cycle Cover} on graphs $G$ in which at least $n-k$ vertices have degree at least $\frac{n}{2}$, generalizing the Jansen, Kozma, and Nederlof's result presented in~\cite{jansen2019hamiltonicity} (WG~2019) for the {\sc Hamiltonian Cycle} problem. Also, the single-exponential FPT algorithm for {\sc Minimum Cycle Cover} parameterized by co-degeneracy implies that {\sc Hamiltonian Cycle} can be solved with the same running time, improving the current state of the art for solving the {\sc Hamiltonian Cycle} problem parameterized by co-degeneracy since the algorithm presented in~\cite{MFCS_2021} runs in $2^{\mathcal{O}(k\log k)}\cdot n^{\mathcal{O}(1)}$ time, where $k$ is the co-degeneracy. 
Note that our results also imply that {\sc Minimum Cycle Cover} on \emph{co-planar} graphs can be solved in polynomial time, which seemed to be unknown in the literature.  

%Finally, we also remark that for some graph problems, co-treewidth can be a parameter more useful than treewidth. For instance, {\sc Equitable Coloring} and {\sc Precoloring Extension} are well-known W[1]-hard problems concerning treewidth; however, we remark that both problems are fixed-parameter tractable when parameterized by co-treewidth.

\section{On the stability of having a bounded cycle cover}

Although the stability of several properties has already been studied (c.f.~\cite{broersma2000closure}), the stability of the property of having a cycle cover of size at most $r$, to the best of our knowledge, is unknown. Therefore, we show that $s(\Upsilon) \leq n$, where $r$ is any positive integer, and $\Upsilon$ is the property of having a cycle cover of size at most $r$.

\begin{lemma}\label{cyclecovertheorem}
Let $r$ be a positive integer.
A simple graph $G$ with $n$ vertices has a cycle cover of size at most $r$ if and only if its $n$-closure, $\operatorname{cl}_{n}(G)$, has also a cycle cover of size at most $r$.
\end{lemma}

\begin{proof}
Let $G$ be a simple graph with $n$ vertices, $r$ be a positive integer, and $\Upsilon$ be the graph property of having a cycle cover of size at most $r$. 
Since the claim trivially holds when $r=0$ or $r\geq n$, we assume that $1\leq r\leq n-1$.

First, note that if $G$ has a cycle cover $S$ of size $r$ then the set $S$ is also a cycle cover of $\operatorname{cl}_{n}(G)$, because $G$ is a spanning subgraph of $\operatorname{cl}_{n}(G)$.

Now, suppose that $G$ does not have a cycle cover of size at most $r$ but $\operatorname{cl}_{n}(G)$ has a cycle cover of size at most $r$.

Given that $\operatorname{cl}_{n}(G)$ is uniquely determined from $G$~\cite{bondy1976method}, the construction of $\operatorname{cl}_{n+\ell}(G)$ can be seen as an iterative process of adding edges, starting from $G$, where a single edge is added at each step $i$, until no more edges can be added.
Let $E_0 = E(G)$. We call by $E_i$ the resulting set of edges after adding $i$ edges during such a process. Therefore, $G_0 = G$, $G_1 = (V, E_1), G_2 = (V, E_2),\ldots, G_t = (V, E_t)$, where $G_t = \operatorname{cl}_{n}(G)$ is the finite sequence of graphs generated during a construction of the $n$-closure of $G$.

Since $G$ does not have a cycle cover of size at most $r$ but $\operatorname{cl}_{n}(G)$ has a cycle cover of size at most $r$, by the construction of $\operatorname{cl}_{n}(G)$, there is a single $i$ ($1 \leq i \leq t$) such that $G_{i-1}$ does not has a cycle cover of size at most $r$ but $G_i$ has a cycle cover of size at most $r$. 
Let $\{uw\} = E_i \setminus E_{i-1}$. 

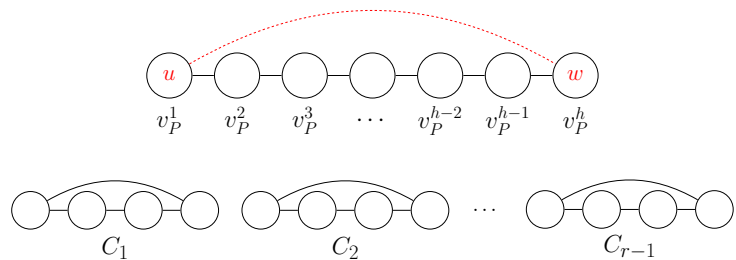
\begin{figure}[!htb]
    \centering
%    \subfloat{%
        \scalebox{0.3}{
            \begin{tikzpicture}
                \begin{scope}[every node/.style={circle,thick,draw, scale=2, minimum size=1cm}]
                    \node (1) at (0,0) {\Large \textcolor{red}{$u$}};
                    \node (2) at (3,0) {\Large };
                    \node (3) at (6,0) {\Large };
                    \node (4) at (9,0) {\Large };
                    \node (5) at (12,0) {};
                    \node (6) at (15,0) {};
                    \node (7) at (18,0) {\Large \textcolor{red}{$w$}};
                \end{scope}
 
                \begin{scope}[>={stealth[black]},
                              every node/.style={fill=white,circle},
                              every edge/.style={draw=black, thick}
                              ]
                    \path [-] (1) edge (2); 
                    \path [-] (2) edge (3);
                    \path [-] (3) edge (4);
                    \path [-] (4) edge (5);
                    \path [-] (5) edge (6);
                    \path [-] (6) edge  (7);
                \end{scope}
                    \begin{scope}[>={stealth[black]},
                              every node/.style={fill=white,circle},
                              every edge/.style={draw=red, thick}
                              ]
                    \path [-] (1) edge [bend left]  [dashed]  (7);
                \end{scope}
                   \begin{scope}[every node/.style={scale=1.5}]
                    \node (label) at (0,-2) {\huge {$v_P^1$}};
                    \node (label) at (3,-2) {\huge {$v_P^2$}};
                    \node (label) at (6,-2) {\huge {$v_P^3$}};
                    \node (label) at (9,-2) {\huge {$\cdots$}};
                    \node (label) at (12,-2) {\huge {$v_P^{h-2}$}};
                    \node (label) at (15,-2) {\huge {$v_P^{h-1}$}};
                    \node (label) at (18,-2) {\huge {$v_P^h$}};
                \end{scope}
            \end{tikzpicture}
%        }
    }%
    
\bigskip

%    \advance\leftskip-2cm
%    \subfloat{
        \centering
        \scalebox{0.25}{
            \begin{tikzpicture}
                \begin{scope}[every node/.style={circle,thick,draw, scale=2, minimum size=1cm}]
                    \node (1) at (0,0) {\Large};
                    \node (2) at (3,0) {\Large };
                    \node (3) at (6,0) {\Large };
                    \node (4) at (9,0) {\Large };
                \end{scope}
 
                \begin{scope}[>={stealth[black]},
                              every node/.style={fill=white,circle},
                              every edge/.style={draw=black, thick}
                              ]
                    \path [-] (1) edge (2); 
                    \path [-] (2) edge (3);
                    \path [-] (3) edge (4);
                    \path [-] (1) edge [bend left]  (4);
                \end{scope}
                   \begin{scope}[every node/.style={scale=1.5}]
                    \node (label) at (4.5,-2) {\Huge $C_1$};
                \end{scope}
            \end{tikzpicture}
            \hspace{1cm}
             \begin{tikzpicture}
                \begin{scope}[every node/.style={circle,thick,draw, scale=2, minimum size=1cm}]
                    \node (1) at (0,0) {\Large};
                    \node (2) at (3,0) {\Large };
                    \node (3) at (6,0) {\Large };
                    \node (4) at (9,0) {\Large };
                \end{scope}
                \begin{scope}[every node/.style={scale=1.5}]
                    \node (label) at (12,0) {\huge {$\cdots$}};
                \end{scope}
 
                \begin{scope}[>={stealth[black]},
                              every node/.style={fill=white,circle},
                              every edge/.style={draw=black, thick}
                              ]
                    \path [-] (1) edge (2); 
                    \path [-] (2) edge (3);
                    \path [-] (3) edge (4);
                    \path [-] (1) edge [bend left]  (4);
                \end{scope}
                   \begin{scope}[every node/.style={scale=1.5}]
                    \node (label) at (4.5,-2) {\Huge $C_2$};
                \end{scope}
            \end{tikzpicture}
            \hspace{1cm}
             \begin{tikzpicture}
                \begin{scope}[every node/.style={circle,thick,draw, scale=2, minimum size=1cm}]
                    \node (1) at (0,0) {\Large};
                    \node (2) at (3,0) {\Large };
                    \node (3) at (6,0) {\Large };
                    \node (4) at (9,0) {\Large };
                \end{scope}
 
                \begin{scope}[>={stealth[black]},
                              every node/.style={fill=white,circle},
                              every edge/.style={draw=black, thick}
                              ]
                    \path [-] (1) edge (2); 
                    \path [-] (2) edge (3);
                    \path [-] (3) edge (4);
                    \path [-] (1) edge [bend left]  (4);
                \end{scope}
                   \begin{scope}[every node/.style={scale=1.5}]
                    \node (label) at (4.5,-2) {\Huge $C_{r-1}$};
                \end{scope}
            \end{tikzpicture}
%        }
    }
    \caption{Representation of a graph with $r-1$ cycles, a path of size $h$ and the edge $vw$ that will be added, creating a graph with $r$ cycles.}
    \label{fig1}
\end{figure}

%\vspace{-.5cm}

Suppose that $G_i$ has a cycle cover $S_i$ of size at most $r$. For simplicity, we assume that $|S_i|=r$.
Therefore, the vertices of $G_{i-1}$ can be covered by a set formed by $r-1$ cycles $C_1, C_2, \ldots, C_{r-1}$ and a path $P$ (the cycle of $G_i$ that contains the edge $uw$). Assume that each cycle $C_j$ is defined by the sequence $v_{C_j}^1,v_{C_j}^2,\ldots, v_{C_j}^{x_j}$ of vertices, where $x_j$ is the number of vertices of $C_j$. Let $C = \{C_1, C_2, \ldots, C_{r-1}\}$, and $P = v_P^1,v_P^2, \ldots, v_P^h$, where $u=v_P^1$, $w=v_P^h$ and $h$ is the number of vertices of $P$. Note that $h\geq 3$; otherwise $u=w$, implying that $P$ is a trivial cycle and $G_{i-1}$ has a cycle cover of size $r$. Figure~\ref{fig1} illustrates $C$ and $P$.

We partition some vertices of $G_{i-1}$ into four sets:

$$X_P = \{v_P^q\ |\ (v_P^{q-1}, v_P^h) \in E_{i-1} \mbox{ and } 2 < q < h \},$$

\vspace{-.5cm}

$$X_C = \{v_{C_j}^q\ |\ (v_{C_j}^{(q \mbox{ mod } x_j) +1}, v_P^h) \in E_{i-1}, \ 1 \leq q \leq x_j,  \mbox{ and } C_j\in C  \},$$

\vspace{-.5cm}

$$Y_P = \{v_P^q\ |\ (v_P^{1}, v_P^q) \in E_{i-1}\ and\ 2 < q < h \},$$

%\vspace{-.2cm}

and 

\vspace{-.5cm}

$$Y_C = \{v_{C_j}^q\ |\ (v_{C_j}^q, v_P^1) \in E_{i-1}, 1 \leq q \leq x_j, \mbox{ and } C_j\in C\}.$$

%\medskip 

Note that $v_{C_j}^1 = v_{C_j}^{(1\mbox{ mod }1) +1}$ for trivial cycles $C_j$. Thus, $X_C$ it is well defined.

Let $X=X_P\cup X_C$ and $Y=Y_P\cup Y_C$.

The set $X$, is the set of vertices (with the exception of $v_P^h$) in which its predecessor in the path or its successor in the cycle is adjacent to $v_P^h$. Also, the set $Y$, is the set of vertices adjacent to $v_P^1$ (with the exception of $v_P^2$). Note that the size of both $X$ and $Y$ are bounded by $n - 3$, since they exclude the vertices $v_P^1,v_P^2$ and $v_P^h$ of $P$. 
Besides that, we can observe that $$|X| = d(v_P^h)-1 \mbox{ and } |Y| = d(v_P^1)-1,$$ where $d(v)$ is the degree of the vertex $v$. Therefore, the following holds:

$$|X| + |Y|  = d(v_P^h) + d(v_P^1) - 2$$
that is,
$$|X| + |Y| \geq n - 2$$
since $d(u)+d(w)\geq n$ where $u=v_P^1$, $w=v_P^h$, and $\{uw\} = E_i \setminus E_{i-1}$.

%\medskip

However, $|X\cup Y|\leq n-3$ because both $X$ and $Y$ exclude $v_P^1,v_P^2$ and $v_P^r$.
Therefore, there is at least one vertex that belong to both $X$ and $Y$. Note that $(X_P\cup Y_P) \cap (X_C\cup Y_C)=\emptyset$, since, by definition, the elements of the covering are vertex disjoint.

Therefore, there are two possibilities:

\begin{enumerate}
    
    \item There is a vertex $v_P^q$ belonging to the path $P$ such that $v_P^q\in X_P\cap Y_P$. This implies that $G_{i-1}$ already had a cycle covering exactly the vertices of $P$ before the addition of the edge $uw=v_P^1v_P^h$, which could be formed as follows (see Fig.~\ref{case1}): $$v_P^1,v_P^2,\ldots, v_P^{q-1}, v_P^h, v_P^{h-1}, v_P^{h-2}, \ldots, v_P^{q+1}, v_P^{q},v_P^1;$$
    
\begin{figure}[!htb]
    \centering
%    \subfloat{%
        \scalebox{0.3}{
            \begin{tikzpicture}
                \begin{scope}[every node/.style={circle,thick,draw, scale=2, minimum size=1cm}]
                    \node (1) at (0,0) {\Large \textcolor{red}{$u$}};
                    \node (2) at (3,0) {\Large };
                    \node (3) at (6,0) {\Large };
                    \node (4) at (9,0) {\Large };
                \end{scope}
                \begin{scope}[every node/.style={fill=lightgray,circle,thick,draw, scale=2, minimum size=1cm}]
                    \node (5) at (12,0) {};
                \end{scope}
                \begin{scope}[every node/.style={circle,thick,draw, scale=2, minimum size=1cm}]
                    \node (6) at (15,0) {};
                    \node (7) at (18,0) {\Large};
                    \node (8) at (21,0) {\Large \textcolor{red}{$w$}};
                \end{scope}
 
                \begin{scope}[>={stealth[black]},
                              every node/.style={fill=white,circle},
                              every edge/.style={draw=black, thick}
                              ]
                    \path [-] (1) edge (2); 
                    \path [-] (2) edge (3);
                    \path [-] (3) edge (4);
                    \path [-] (4) edge (5);
                    \path [-] (5) edge (6);
                    \path [-] (6) edge  (7);
                    \path [-] (7) edge  (8);
                \end{scope}
                    \begin{scope}[>={stealth[black]},
                              every node/.style={fill=white,circle},
                              every edge/.style={draw=blue, thick}
                              ]
                            \path [-] (1) edge [bend left]  [dashed]  (5);
                    \end{scope}
                    \begin{scope}[>={stealth[black]},
                              every node/.style={fill=white,circle},
                              every edge/.style={draw=blue, thick}
                              ]
                              \path [-] (8) edge [bend left]  [dashed]  (4);
                    \end{scope}                
                   \begin{scope}[every node/.style={scale=1.5}]
                    \node (label) at (0,-3) {\huge {$v_P^1$}};
                    \node (label) at (3,-3) {\huge {$v_P^2$}};
                    \node (label) at (6,-3) {\huge {$\cdots$}};
                    \node (label) at (9,-3) {\huge {$v_P^{q-1}$}};
                    \node (label) at (12,-3) {\huge {$v_P^{q}$}};
                    \node (label) at (15,-3) {\huge {$\cdots$}};
                    \node (label) at (18,-3) {\huge {$v_P^{h-1}$}};
                    \node (label) at (21,-3) {\huge {$v_P^{h}$}};
                \end{scope}
            \end{tikzpicture}
%        }
    }%
    \caption{Representation of case 1, where the vertex $v_P^q$, highlighted in gray, belongs to $X_P\cap Y_P$.}
    \label{case1}
\end{figure}
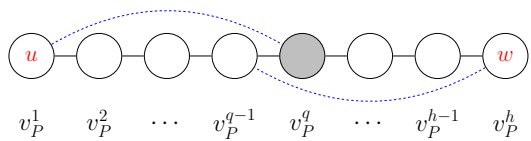
    
    \item There is a vertex $v_{C_j}^q$ belonging to a cycle $C_j\in C$ such that $v_{C_j}^q\in X_C\cap Y_C$. In this case, $G_{i-1}$ has a larger cycle that can be obtained by merging $C_j$ with the path $P$ as follows (see Fig.\ref{case2}):
     
     $$v_{C_j}^{(q \mbox{ mod } x_j) +1},v_{C_j}^{(q \mbox{ mod } x_j) +2},\ldots,v_{C_j}^{q},v_P^1,v_P^2,\ldots,v_P^h,v_{C_j}^{(q \mbox{ mod } x_j) +1}.$$
     
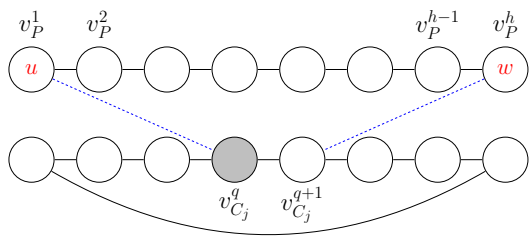
\begin{figure}[!htb]
    \centering
%    \subfloat{%
        \scalebox{0.3}{
            \begin{tikzpicture}
                \begin{scope}[every node/.style={circle,thick,draw, scale=2, minimum size=1cm}]
                    \node (1) at (0,0) {\Large \textcolor{red}{$u$}};
                    \node (2) at (3,0) {\Large };
                    \node (3) at (6,0) {\Large };
                    \node (4) at (9,0) {\Large };
                    \node (5) at (12,0) {};
                    \node (6) at (15,0) {};
                    \node (7) at (18,0) {\Large};
                    \node (8) at (21,0) {\Large \textcolor{red}{$w$}};

                    \node (9) at (0,-4) {};                    
                    \node (10) at (3,-4) { };
                    \node (11) at (6,-4) { };

                    \node (13) at (12,-4) {};
                    \node (14) at (15,-4) {};
                    \node (15) at (18,-4) {};
                    \node (16) at (21,-4) {};
                \end{scope}
                \begin{scope}[every node/.style={fill=lightgray,circle,thick,draw, scale=2, minimum size=1cm}]
                    \node (12) at (9,-4) { };
                \end{scope}
 
                \begin{scope}[>={stealth[black]},
                              every node/.style={fill=white,circle},
                              every edge/.style={draw=black, thick}
                              ]
                    \path [-] (1) edge (2); 
                    \path [-] (2) edge (3);
                    \path [-] (3) edge (4);
                    \path [-] (4) edge (5);
                    \path [-] (5) edge (6);
                    \path [-] (6) edge  (7);
                    \path [-] (7) edge  (8);
                    \path [-] (9) edge  (10);
                    \path [-] (10) edge  (11);
                    \path [-] (11) edge  (12);
                    \path [-] (12) edge  (13);
                    \path [-] (13) edge  (14);
                    \path [-] (14) edge  (15);
                    \path [-] (15) edge  (16);
                    \path [-] (16) edge [bend left]  (9);
                \end{scope}
                   \begin{scope}[>={stealth[black]},
                              every node/.style={fill=white,circle},
                              every edge/.style={draw=blue, thick}
                              ]
                    \path [-] (1) edge  [dashed]  (12);
                \end{scope}
                   \begin{scope}[>={stealth[black]},
                              every node/.style={fill=white,circle},
                              every edge/.style={draw=blue, thick}
                              ]
                    \path [-] (8) edge  [dashed]  (13);
                \end{scope}
               \begin{scope}[every node/.style={scale=1.5}]
                    \node (label) at (0,2) {\huge {$v_P^1$}};
                    \node (label) at (3,2) {\huge {$v_P^2$}};
                    \node (label) at (18,2) {\huge {$v_P^{h-1}$}};
                    \node (label) at (21,2) {\huge {$v_P^{h}$}};
                    \node (label) at (9,-6) {\huge {$v_{C_j}^{q}$}};
                    \node (label) at (12,-6) {\huge {$v_{C_j}^{q+1}$}};
                \end{scope}
            \end{tikzpicture}
%        }
    }%
    \caption{Representation of case 2, where the vertex $v_{C_j}^q$, highlighted in gray, belongs to $X_C\cap Y_C$.}
    \label{case2}
\end{figure}
     
\end{enumerate}     

In the first case $G_{i-1}$ has a cycle cover of size $r$, while in the second case $G_{i-1}$ has a cycle cover of size $r-1$. Both cases contradicts the hypothesis that $G_{i-1}$ does not have a cycle cover of size at most $r$.

Therefore, there is no $1 \leq i\leq t$ such that $G_{i-1}$ does not have a cycle cover of size at most $r$ and $G_i$ has such a cycle cover. Thus, if $G_t=\operatorname{cl}_{n}(G)$ has a cycle cover of size at most $r$ then $G_0=G$ also has a cycle cover of size at most $r$. \qed
\end{proof}

Lemma~\ref{cyclecovertheorem} states that for any positive integer~$r$, the graph property $\Upsilon$ of having a cycle cover of size at most $r$ satisfies that $s(\Upsilon)\leq n$. We remark that such a bound is tight since whenever $r=1$, the target $\Upsilon$ is the Hamiltonicity property, which is well known to have stability (exactly) equal to $n$ (c.f.~\cite{bondy1976method}).

\smallskip

Now, observe that the class of graphs with co-degeneracy at most $k$ is closed under completion (edge addition), in the same way as the class of graphs with degeneracy at most $k$ is closed under edge removals. 
Recall that $\operatorname{cl}_n(G)$ is uniquely determined from a $n$-vertex graph $G$ and it can be constructed in polynomial time. 
Therefore, by Lemma~\ref{cyclecovertheorem}, we may assume that $G=\operatorname{cl}_n(G)$ whenever $G$ is an instance of {\sc Minimum Cycle Cover} parameterized by co-degeneracy.

We call by \emph{co-vertex cover} any set of vertices whose removal makes the resulting graph complete, i.e., a vertex cover in the complement graph. The \emph{co-vertex cover number} of a graph $G$, $co\mbox{-}vc(G)$, is the size of its minimum co-vertex cover.\footnote{$co\mbox{-}vc(G)$ is also called the {\em distance to clique} of $G$, and a co-vertex cover set is also called a {\em clique modulator}.} 

The following theorem is a key tool for this work.

\begin{theorem}[\cite{MFCS_2021}]\label{theo:closure-co-vertexcover}
Let $\ell\geq 0$ be an integer.  If a graph $G$ has co-degeneracy $k$ and $G=\operatorname{cl}_{n+\ell}(G)$ then $G$ has  co-vertex cover number bounded by $2k+\ell+1$.  In addition,  a co-vertex cover of $G$ with size at most $2k+\ell+1$ can be found in polynomial time.
\end{theorem}

From Lemma~\ref{cyclecovertheorem} and Theorem~\ref{theo:closure-co-vertexcover}, the problem of solving {\sc Minimum Cycle Cover} on instances $G$ with co-degeneracy $k$ can be reduced in polynomial time to the problem of solving {\sc Minimum Cycle Cover} on instances $G'=\operatorname{cl}_n(G)$ with co-vertex cover number at most $2k+1$.
Therefore, in the next section we will focus on parameterization by the co-vertex cover number.

\section{Polynomial kernelization}

In~\cite{jansen2019hamiltonicity}, Jansen, Kozma, and Nederlof showed that given a graph $G$ with $n$ vertices such that at least $n-k$ vertices of $G$ have degree at least $\frac{n}{2}$, there is a deterministic algorithm that constructs in polynomial time a graph $G'$ with at most $3k$ vertices, such that $G$ is Hamiltonian if and only if $G'$ is Hamiltonian. In other words, they showed that the {\sc Hamiltonian Cycle}
problem parameterized by such a $k$ has a kernel with a linear number of vertices.

First, we remark that such a parameterization that aims to explore a ``distance measure'' ($k$) of a given graph $G$ from satisfying the Dirac property, when applied to problems that are $n$-stable (such as {\sc Hamiltonian Cycle} and {\sc Minimum Cycle Cover}) can be polynomial-time reduced to the case where the co-degeneracy is bounded by $k$. Since for such problems one can consider only instances $G'$ such that $G'=\operatorname{cl}_n(G')$, from a graph $G$ with $n$ vertices such that at least $n-k$ vertices of $G$ have degree at least $\frac{n}{2}$, we obtain an instance $G'=\operatorname{cl}_n(G)$ having a clique of size at least $n-k$. 

Therefore, in the following, we extend the ``relaxed'' Dirac result from~\cite{jansen2019hamiltonicity} by considering co-degeneracy and the {\sc Minimum Cycle Cover} problem.

%%%%%%%%%%%%%%%%%%%%%%%%%%%%%%%%%%%%%%%%%%%%%%%%%%%%%%%%%%%%%%%%%%%%%

\begin{theorem}\label{theorembartADAPTADO}
There is a polynomial-time algorithm that, given a graph $G$ and a nonempty set $S \subseteq V(G)$ such that $G - S$ is a clique, outputs an induced subgraph $G'$ of $G$ on at most $3|S|$ vertices such that $G$ has a cycle cover of size at most $r$ if and only if $G'$ has a cycle cover of size at most $r$.
\end{theorem}
\begin{proof}
Let $G=(V,E)$ be a graph having a co-vertex cover $S$. Let $C= V(G) \setminus S$. If $|C| \leq 2|S|$ then by setting $G' = G$ the claim holds. Now, assume that $|C| > 2|S|$.

As in~\cite{jansen2019hamiltonicity}, let $S'= \{v_1, v_2 : v \in S\}$ be a set containing two representatives for each vertex of $S$. We construct a bipartite graph $H$ on vertex set $C \cup S'$, where for each edge $cv \in E(G)$ with $c \in C$ and $v \in S$, we add the edges $cv_1$, $cv_2$ to $E(H)$. 

Now, we compute a maximum matching $M \subseteq E(H)$ of $H$. Let $C^*$ be the subset of vertices of $C$ saturated (matched) by $M$. If $|C^*| \geq |S| + 1$ then set $C' = C^*$; otherwise, let $C' \subseteq C$
be a superset of $C^*$ with size $|S| + 1$. Finally, set $G' = G[C' \cup S]$.
%%%%%%%%%%%%%%

Note that $G'$ has at most $3|S|$ vertices, because $C'$ has at most $2|S|$ vertices.

%%%%%%%%%%%%%%%%%
First, suppose that $G'$ has cycle cover $Q'$ of size at most $r$. 
Since $G'$ is a subgraph of $G$, the set $Q'$ is a set of vertex disjoint cycles of $G$ covering $S\cup C'\subseteq V(G)$. Thus, only vertices of $C\setminus C'$ are not covered by $Q'$.  
However, since the size of $C'$ is greater than the size of $S$, there is at least one cycle $Q_j\in Q'$ that either is a single vertex of $C$ or contains an edge between vertices of $C$. If $|Q_j|=1$ then we can replace it by a cycle containing all the vertices of $(C\setminus C')\cup Q_j$. If $Q_j$ has an edge $uv$ such that $u,v\in C$, then we can replace this edge by a $uv$-path containing the vertices of $C\setminus C'$ as internal vertices. In both cases we obtain a cycle cover of size at most $|Q'|$ in the graph $G$.
%%%%%%%%%%%%%%%%%

At this point, it remains to show that if $G$ has a cycle cover of size $r$ then $G'$ has a cycle cover of size at most $r$. 

%Estamos buscando um ciclo cover de tamanho no máximo $k$. Logo, podemos assumir que estamos buscando o mínimo.

%- Ciclos totalmente em S;

%- ciclos mistos com arestas de C;  (só tem 1)

%- ciclos mistos sem arestas de C;

%\medskip

Using a strategy similar to that in~\cite{jansen2019hamiltonicity}, we first present a structure that implies cycle covers of size at most $r$ in $G'$. For a vertex set $S^*$ in a graph $G^*$, we define a {\em cycle-path cover} of $S^*$ in $G^*$ as a set $L$ of pairwise vertex-disjoint simple paths or cycles such that each vertex of $S^*$ belongs to exactly one element of $L$, i.e., $L$ can be seen as a subgraph with maximum degree two which contains every vertex of $S^*$.
For a vertex set $C^*$ in $G^*$, we say that a cycle-path cover $L$ has $C^*$-endpoints if the endpoints of each path $P\in L$ belong to $C^*$. 

%\begin{claim}\label{claimReductionSize}
%Graph $G'$ has at most $3|S|$ vertices.
%\end{claim}
%\begin{proof}
%Since each vertex of $C*$ is matched to a distinct vertex in $S'$, with $|S| = 2|S|$, it follows that $|C*| \leq 2|S|$ which implies $|C'| \leq 2|S|$. As $V(G') = C' \cup S$, the claim follows. 
%\end{proof}

\begin{Claim}\label{claimCyclePath}
If $G'$ has a cycle-path cover of $S$ having $C'$-endpoints and containing at most $r-1$ cycles, then $G'$ has a cycle cover of size at most $r$.
\end{Claim}
\begin{proof}

We have two cases to analyse:
%\begin{enumerate}
%    \item If the cycle/path cover contains just cycles and have an edge $uv$ such that $u,v\in C'$, we can replace it by a $uv$-path containing as internal vertices all the vertices of $G'$ that are not in such a cycle/path cover of $S$. Therefore, $G'$ has a cycle cover of size at most $r$.
%    
%    \item 
    if the cycle-path cover of $S$ contains only cycles, % and do not have an edge between vertices of  $C'$, 
    as the number of cycles is at most $r-1$, then we can add a new cycle formed by the vertices not yet covered; %Thus, $G'$ has a cycle cover of size at most $r$;
    if the cycle-path cover contains some paths, by vertex disjointness, all the paths have different endpoints in $C'$, and, since $C'$ is a clique, we can connect such endpoints in such a way as to form a single cycle containing these paths as subgraphs, after that, an edge $uv$ of such a cycle having $u,v\in C'$ can replaced by a $uv$-path containing as internal vertices the vertices of $G'$ that are not in such a cycle-path cover of $S$. In both cases, we conclude that $G'$ has a cycle cover of size at most $r$. \quad $\lrcorner$
%    
%\end{enumerate}
\end{proof}

%%%%%%%%%%%%%

Now, considering the bipartite graph $H$ and its maximum matching $M$, let $U_C$ be the set of vertices of $C$ that are not saturated by $M$, and let $R$ be the vertices of $H$ that are reachable from $U_C$ by an $M$-alternating path in $H$ (which starts with a non-matching edge). Set $R_C = R \cap C$ and $R_{S'} = R \cap S'$.

By Claim~\ref{claimCyclePath}, it is enough to show that if $G$ has a cycle cover of size $r$ then $G'$ has a cycle-path cover of $S$ having $C'$-endpoints and containing at most $r-1$ cycles. For that, we consider Claim~\ref{claim35Bart} presented in~\cite{jansen2019hamiltonicity}.

\begin{Claim}[\cite{jansen2019hamiltonicity}]\label{claim35Bart}
The sets $R$, $R_C$, $R_{S'}$ satisfy the following.

\begin{enumerate}
    \item Each $M$-alternating path in $H$ from $U_C$ to a vertex in $R_{S'}$ (resp. $R_C$) ends with a non-matching (resp. matching) edge.

    \item Each vertex of $R_{S'}$ is matched by $M$ to a vertex in $R_C$.
    
    \item For each vertex $x \in R_C$ we have $N_H(x) \subseteq R_{S'}$.
    
    \item For each vertex $v \in S$ we have $v_1 \in R_{S'}$ if and only if $v_2 \in R_{S'}$.
    
    \item For each vertex $v \in S' \setminus R_{S'}$ , we have $N_H(v) \cap R_C = \emptyset$ and each vertex of $N_H(v)$ is saturated by $M$.
\end{enumerate}

\end{Claim}
%%%%%%%%%%%

\begin{lemma}\label{FromGtopath-cyclecover}
If $G$ has a cycle cover of size at most $r$, then $G'$ has a cycle-path cover of $S$ having $C'$-endpoints and containing at most $r-1$ cycles. 
%$G'$ has a cycle cover of size at most $k$.
\end{lemma}
\begin{proof}
Let $F$ be a cycle cover of size at most $r$ of $G$. 
%By Claim~\ref{claimCyclePath} it suffices to build a cycle/path cover of $S$ in $G'$ with $C'$-endpoints. 
Consider $F$ as a 2-regular subgraph of $G$. 
Let $F_1 = F[S]$ be the subgraph of $F$ induced by $S$. Since $F$ is a spanning subgraph of $G$, and $S\subset V(G')$, it follows that $F_1$ is a cycle-path cover of $S$ in $G'$. At this point, we need to extend it to have $C'$-endpoints. %However, the paths in $F_1$ have their endpoints in $S$ rather than in $C'$. 
As in~\cite{jansen2019hamiltonicity}, we do that by inserting edges into $F_1$ to turn it into a subgraph $F_2$ of $G'$ in which each vertex of $S$ has degree exactly two.
%
%O cuidado principal é evitar criar mais ciclos do que k-1
%
This structure $F_2$ must be a cycle-path cover of $S$ in $G'$ with $C'$-endpoints, since the degree-two vertices $S$ cannot be endpoints of the paths. 

Setting $F_2=F_1$, $R_S = \{v \in S : v_1 \in R_{S'} \mbox{ or } v_2 \in R_{S'}\}$,
%
%O cuidado principal é evitar criar mais ciclos do que k-1
%
we~proceed~as~follows.

\begin{enumerate}
    \item For each vertex $v \in R_S$, we have $v_1,\ v_2 \in R_{S'}$ by Claim~\ref{claim35Bart}(4), which implies by Claim~\ref{claim35Bart}(2) that both $v_1$ and $v_2$ are matched to distinct vertices $x_1$, $x_2$ in $R_C$. If $v$ has degree zero in subgraph $F_1$, then add the edges $vx_1$, $vx_2$ to $F_2$. If $v$ has degree one in $F_2$ then only add the edge $vx_1$. (we do not add edges if $v$ already has degree two in $F_1$)

    \item For each vertex $v \in S \setminus R_S$, it holds that $N_G(v) \cap R_C = \emptyset$. This follows from the fact that $N_G(v)\ =\ N_H(v_1)\ =\ N_H(v_2)$  and Claim~\ref{claim35Bart}(5). Note that $v \notin R_S$ implies $v_1,\ v_2 \notin R_{S'}$. 
    Hence the (up to two) neighbors that $v \in S \setminus R_S$ has in $C$ on the cycle cover $F$ do not belong to $R_C$ (see also Claim~\ref{claim35Bart}(3)), 
    In addition, Claim~\ref{claim35Bart}(5) ensures that all vertices of $N_G(v)$ are saturated by $H$ and hence belong to $C'$. 
    Thus, for each vertex $v \in S \setminus R_S$, for each edge from $v$ to $C \cap C'$ incident on $v$ in $F$, we insert the corresponding edge into $F_2$.
\end{enumerate}

It is clear that the above procedure produces a subgraph $F_2$ in which all vertices of $S$ have degree exactly two. 
%%%%
By Claim~\ref{claim35Bart}(5), we have that a vertex $c \in C$ does not have edges added in $F_2$ by both previous steps, thus each vertex $c\in C$ added in $F_2$ has degree at most two in it because 
$c$ has at most one edge in the matching $M$ (see Step 1), while 
$c$ has two edges in the cycle cover $F$ (see Step 2).
%the edges inserted for $v\in R_S$ connect to distinct vertices in $C' \cap R_C$, while the edges inserted for $v\in S \setminus R_S$ connect to $C' \setminus R_C$ in the same way as in the cycle cover $F$. 

At this point, we know that $F_2$ is a cycle-path cover of $S$ having $C'$-endpoints. It remains to show that it contains at most $r-1$ cycles. 

\begin{Claim}\label{claimCycleF2}
Every cycle of $F_2$ is a cycle of $F$.
\end{Claim}
\begin{proof}
Suppose that $F_2$ has a cycle $Q$ that is not in $F$. As $F_2$ is formed from $F_1$, the edges in $Q$ between the vertices of $S$ are also edges of $F$. Furthermore, by construction, the added edges from $F_1$ to obtain $F_2$ are the edges incident to the vertices of $S$. Therefore, there is no edge between the vertices of the clique $C$ in $Q$. 
By Claim~\ref{claim35Bart}(5), we have that a vertex $c \in C$ cannot be incident to two edges of $F_2$ being one added by Step 1 and the other by Step 2 of the construction. 
Since these steps are mutually exclusive with respect to a vertex $c \in C$, and given that $c$ has degree two in $Q$ (since $Q$ is a cycle), we have that the edges of each vertex $c \in C\cap Q$ were added by Step 2 of the construction (Step 1 adds only one edge of the matching). 
However, by construction, the edges in $Q$ incident to a vertex $c \in C$ are the edges in $F$. Therefore, every edge of $Q$ is contained in $F$, contradicting the hypothesis that $Q$ is not contained in $F$. \quad $\lrcorner$
\end{proof}

By hypothesis, $F$ has at most $r$ cycles. Since $|C|>|S|$, it holds that at least one cycle of $F$ must have an edge between vertices of $C$. Thus, at least one cycle of $F$ is not completely contained in $F_2$, which implies, by Claim~\ref{claimCycleF2}, that $F_2$ has at most $r-1$ cycles. Therefore, $F_2$ is a cycle-path cover of $S$ having $C'$-endpoints which contains at most $r-1$ cycles. This concludes the proof of Lemma~\ref{FromGtopath-cyclecover}. \quad $\lrcorner$
\end{proof}

By Lemma~\ref{FromGtopath-cyclecover} and Claim \ref{claimCyclePath}, it holds that if $G$ has a cycle cover of size at most $r$ then $G'$ has a cycle cover of size at most $r$. Since the reduction can be performed in polynomial time, and $|V(G')|\leq 3|S|$, we conclude the proof of Theorem~\ref{theorembartADAPTADO}. \qed
\end{proof}

%%%%%%%%%%%%%%%%%%%%%%%%%%%%%%%%%%%%%%% %%%%%%%%%%%%%%%%%%%%%%%%%%%%%

\begin{corollary}\label{cyclecoverKernel}
{\sc Minimum Cycle Cover} parameterized by co-degeneracy admits a kernel with at most $6k+3$ vertices, where $k=co$-$deg$.
\end{corollary}
%\begin{proof}
%It follows from Lemma~\ref{cyclecovertheorem}, Theorem~\ref{theo:closure-co-vertexcover}, and Theorem~\ref{theorembartADAPTADO}.  
%\end{proof}

\section{An exact single-exponential time algorithm}

By Corollary~\ref{cyclecoverKernel}, it holds that an exact and deterministic single-exponential time algorithm for {\sc Minimum Cycle Cover} is enough to obtain an FPT algorithm for {\sc Minimum Cycle Cover} with single-exponential dependency concerning the co-degeneracy of the input graph.  
%However, we did not find in the literature an explicit algorithm to find an optimal cycle cover deterministically and with a single-exponential running time. 
In~\cite{cutandcount}, using the Cut\&Count technique, M. Cygan, J. Nederlof, Ma. Pilipczuk, Mi. Pilipczuk, J. Rooij and J. Wojtaszczyk produces a $2^{\mathcal{O}(tw)}\cdot|V|^{\mathcal{O}(1)}$ time Monte Carlo
algorithm for {\sc Minimum Cycle Cover} ({\sc Undirected Min Cycle Cover} in~\cite{cutandcount}), where $tw$ is the treewidth of the input graph. 
In~\cite{BODLAENDER201586}, H. Bodlaender, M. Cygan, S. Kratsch, J. Nederlof presented two approaches to design deterministic $2^{\mathcal{O}(tw)}\cdot|V|^{\mathcal{O}(1)}$-time algorithms for some connectivity problems, and claimed that such approaches can be apply to all problems studied in~\cite{cutandcount}. 

Although such approaches can be used to solve {\sc Minimum Cycle Cover} by a single-exponential time algorithm, in order to present a simpler deterministic procedure, below we present a simple and deterministic dynamic programming based on modifying the Bellman–Held–Karp algorithm. 

%Now, we want to show that we can make a single-exponential algorithm for the {\sc Cycle Cover} parameterized by the $co$-$deg$.

\begin{theorem}\label{single-exp-alg}
{\sc Minimum Cycle Cover} can be solved in $\mathcal{O}(2^n\cdot n^3)$ time.
\end{theorem}

\begin{proof}
Given a graph $G=(V,E)$ with an isolated vertex $w$, a vertex subset $X\subseteq V$, $s,t\in X$, and a Boolean variable $P2$, we denote by 
$M[X,s,t,P2]$ 
the size of a minimum set $S$ of vertex-disjoint cycles but one nonempty vertex-disjoint $st$-path of $G[X]$ such that 
\begin{itemize}
    \item every vertex of $X$ is in an element of $S$;
    \item the $st$-path is not a $P_2$ if the variable $P2=0$;
    \item the $st$-path is a $P_2$ if the variable $P2=1$.
\end{itemize}

Note that $M[V,w,w,0]$ represents the size of a minimum cycle cover of $G$. 

In essence, the $st$-path represents the open cycle that is still being built. The variable $P2$ is a control variable to avoid $P_2$ as cycles of size two. At each step, we can interpret that the algorithm either lengthens the path by adding a new endpoint or closes a cycle and opens a new trivial path. As we can reduce the {\sc Minimum Cycle Cover} problem to the case where the graph has an isolated vertex $w$, we assume that this is the case and consider $w\in X$ just when $X=V$. 

Our recurrence is as follows.

If $X=\{v\}$ then $M[X, s, t, P2]=1$ for $s=t=v$ and $P2=0$;

otherwise, it is $\infty$.

\medskip

If $|X|\geq 2$ then $M[X, s, t, P2]$ is equal to

\[
\begin{cases}
\infty & \text{if $s = t$, $P2 = 1$ }\\

\min\limits_{s',t'\in X\setminus \{t\}\ :\ s't'\in E \mbox{ or } s'=t'}(M[X \setminus \{t\}, s', t', 0]) + 1 & \text{if $s = t$, $P2 = 0$ }\\

\infty & \text{if $s\neq t$, $P2 = 1$, $st\notin E$ }\\

M[X \setminus \{t\}, s, s, 0] & \text{if $s\neq t$, $P2 = 1$, $st\in E$ }\\

\min\limits_{t'\in X\setminus \{t\}\ s.t.\ tt'\in E,\ P_2'\in\{0,1\}}(M[X \setminus \{t\}, s, t', P2']) & \text{if $s\neq t$, $P2 = 0$ }
\end{cases}
\]

%\forall t'\in X\setminus \{t\} : tt'\in E  

%Algorithm~\ref{cyclecoveralg} illustrates our dynamic programming algorithm based on the above recurrence relation.

%Using a table $DP[used, first, last, pathSize]$, where $used$ is the set of vertices that already taken in a closed cycle or in the current path, $first$ is the first vertex of the current path, $last$ is the last vertex of the current path and $pathSize$ is the size of the current path, for our purpose we just need to know if the path has $1$, $2$ or more vertices, so $pathSize$ is limited to $3$. 

%\newpage

The size of the table is bounded by $(2^n-1) \cdot n^2 \cdot 2$
%  +1$ (where the plus one comes from the only case for the vertex $w$, i.e., $X=V,s=t=w$, and $P2=0$), 
where $n$ is the number of vertices of the graph. Regarding time complexity, we have three cases: when $P2=1$ the recurrences can be computed in $\mathcal{O}(1)$ time; when $s=t$ and $P2=0$ the recurrence can be computed in $\mathcal{O}(n^2)$ time, and since there are at most $(2^n-1)\cdot n+1$ cells in this case, the total amount of time taken to compute those cells is $\mathcal{O}(2^n\cdot n^3)$; finally, when $s\neq t$ and $P2=0$ the recurrence can be computed in $\mathcal{O}(n)$ time, but there are $\mathcal{O}(2^n\cdot n^2)$ cells in this case, implying into a total amount of $\mathcal{O}(2^n\cdot n^3)$ time to compute all these cells. Therefore, the dynamic programming algorithm can be performed in $\mathcal{O}(2^n\cdot n^3)$ time. Note that, in addition to determining the size of a minimum cycle cover, one can find it with the same running time. Also, the correctness of the algorithm is straightforward. \qed
%This completes the proof. 

\end{proof}

%Note that To resolve each cell, the time complexity is $\mathcal{O}(n)$. So, the time complexity for this algorithm is $\mathcal{O}(2^n \cdot n^3 \cdot 3)$. To get the answer we just need to initialize the table of the dynamic programming with $-1$ and call the procedure $CycleCover(\{v\}, v, v, 0)$, where $v$ is any vertex of $G$.

\begin{corollary}\label{cyclecoverCor}
{\sc Minimum Cycle Cover} can be solved in $2^{\mathcal{O}(co\mbox{-}deg)}\cdot n^{\mathcal{O}(1)}$ time.
\end{corollary}
%\begin{proof}
%It follows from Corollary~\ref{cyclecoverKernel} and Theorem~\ref{single-exp-alg}.
%\end{proof}

By Corollary~\ref{cyclecoverCor}, it follows that {\sc Minimum Cycle Cover} on co-planar graphs can be solved in polynomial time, which seems to be unknown in the literature.  

\begin{corollary}
{\sc Minimum Cycle Cover} on graphs $G$ in which at least $n-k$ vertices have degree at least $\frac{n}{2}$ can be solved in $2^{\mathcal{O}(k)}\cdot n^{\mathcal{O}(1)}$ time.
\end{corollary}
%\begin{proof}
%It follows from Lemma~\ref{cyclecovertheorem} and Corollary~\ref{cyclecoverCor}.
%\end{proof}
%%%%%%%%%%%%%%%%%%%%%%%%%%%%%%%%%%%%%%%%%%%%%%

%\bibliographystyle{spbasic}      % basic style, author-year citations
%\bibliographystyle{spmpsci}      % mathematics and physical sciences
\bibliographystyle{abbrv}       % APS-like style for physics
\bibliography{refs2} 

\begin{thebibliography}{10}

\bibitem{BODLAENDER201586}
H.~L. Bodlaender, M.~Cygan, S.~Kratsch, and J.~Nederlof.
\newblock Deterministic single exponential time algorithms for connectivity
  problems parameterized by treewidth.
\newblock {\em Information and Computation}, 243:86--111, 2015.
\newblock 40th International Colloquium on Automata, Languages and Programming
  (ICALP 2013).

\bibitem{bondy1976method}
J.~A. Bondy and V.~Chv{\'a}tal.
\newblock A method in graph theory.
\newblock {\em Discrete Mathematics}, 15(2):111--135, 1976.

\bibitem{broersma2000closure}
H.~Broersma, Z.~Ryj{\'a}{\v{c}}ek, and I.~Schiermeyer.
\newblock Closure concepts: a survey.
\newblock {\em Graphs and Combinatorics}, 16(1):17--48, 2000.

\bibitem{corneil2005relationship}
D.~G. Corneil and U.~Rotics.
\newblock On the relationship between clique-width and treewidth.
\newblock {\em SIAM Journal on Computing}, 34(4):825--847, 2005.

\bibitem{C90}
B.~Courcelle.
\newblock The monadic second-order logic of graphs. {I.} recognizable sets of
  finite graphs.
\newblock {\em Information and Computation}, 85(1):12--75, 1990.

\bibitem{courcelle1994monadic}
B.~Courcelle.
\newblock The monadic second order logic of graphs {VI}: On several
  representations of graphs by relational structures.
\newblock {\em Discrete Applied Mathematics}, 54(2-3):117--149, 1994.

\bibitem{Co00}
B.~Courcelle, J.~A. Makowsky, and U.~Rotics.
\newblock Linear time solvable optimization problems on graphs of bounded
  clique-width.
\newblock {\em Theory of Computing Systems}, 33(2):125--150, 2000.

\bibitem{CO00-2}
B.~Courcelle and S.~Olariu.
\newblock Upper bounds to the clique width of graphs.
\newblock {\em Discrete Applied Mathematics}, 101(1--3):77--114, 2000.

\bibitem{cutandcount}
M.~Cygan, J.~Nederlof, M.~Pilipczuk, M.~Pilipczuk, J.~M.~M. Van~Rooij, and
  J.~O. Wojtaszczyk.
\newblock Solving connectivity problems parameterized by treewidth in single
  exponential time.
\newblock {\em ACM Trans. Algorithms}, 18(2), mar 2022.

\bibitem{dirac1952some}
G.~A. Dirac.
\newblock Some theorems on abstract graphs.
\newblock {\em Proceedings of the London Mathematical Society}, 3(1):69--81,
  1952.

\bibitem{MFCS_2021}
G.~L. Duarte, M.~de~Oliveira~Oliveira, and U.~S. Souza.
\newblock {Co-Degeneracy and Co-Treewidth: Using the Complement to Solve Dense
  Instances}.
\newblock In F.~Bonchi and S.~J. Puglisi, editors, {\em 46th International
  Symposium on Mathematical Foundations of Computer Science (MFCS 2021)},
  volume 202 of {\em Leibniz International Proceedings in Informatics
  (LIPIcs)}, pages 42:1--42:17, Dagstuhl, Germany, 2021. Schloss Dagstuhl --
  Leibniz-Zentrum f{\"u}r Informatik.

\bibitem{Bond}
G.~L. Duarte, H.~Eto, T.~Hanaka, Y.~Kobayashi, Y.~Kobayashi, D.~Lokshtanov,
  L.~L.~C. Pedrosa, R.~C.~S. Schouery, and U.~S. Souza.
\newblock Computing the largest bond and the maximum connected cut of a graph.
\newblock {\em Algorithmica}, 83(5):1421--1458, 2021.

\bibitem{DBLP:journals/corr/DvorakKM16a}
P.~Dvoř{\'{a}}k, D.~Knop, and T.~Masar{\'{\i}}k.
\newblock Anti-path cover on sparse graph classes.
\newblock In J.~Bouda, L.~Hol{\'{\i}}k, J.~Kofron, J.~Strejcek, and
  A.~Rambousek, editors, {\em Proceedings 11th Doctoral Workshop on
  Mathematical and Engineering Methods in Computer Science, {MEMICS} 2016,
  Tel{\v{c}}, Czech Republic, 21st-23rd October 2016}, volume 233 of {\em
  {EPTCS}}, pages 82--86, 2016.

\bibitem{SJDM_fahad}
F.~V. Fomin, P.~A. Golovach, D.~Lokshtanov, F.~Panolan, S.~Saurabh, and
  M.~Zehavi.
\newblock Going far from degeneracy.
\newblock {\em SIAM Journal on Discrete Mathematics}, 34(3):1587--1601, 2020.

\bibitem{fomin2009clique}
F.~V. Fomin, P.~A. Golovach, D.~Lokshtanov, and S.~Saurabh.
\newblock Clique-width: on the price of generality.
\newblock In {\em Proceedings of the twentieth annual ACM-SIAM symposium on
  Discrete algorithms}, pages 825--834. SIAM, 2009.

\bibitem{fomin2010algorithmic}
F.~V. Fomin, P.~A. Golovach, D.~Lokshtanov, and S.~Saurabh.
\newblock Algorithmic lower bounds for problems parameterized by clique-width.
\newblock In {\em Proceedings of the twenty-first annual ACM-SIAM symposium on
  Discrete Algorithms}, pages 493--502. SIAM, 2010.

\bibitem{fomin2010intractability}
F.~V. Fomin, P.~A. Golovach, D.~Lokshtanov, and S.~Saurabh.
\newblock Intractability of clique-width parameterizations.
\newblock {\em SIAM Journal on Computing}, 39(5):1941--1956, 2010.

\bibitem{fomin2014almost}
F.~V. Fomin, P.~A. Golovach, D.~Lokshtanov, and S.~Saurabh.
\newblock Almost optimal lower bounds for problems parameterized by
  clique-width.
\newblock {\em SIAM Journal on Computing}, 43(5):1541--1563, 2014.

\bibitem{SODA_Fedor}
F.~V. Fomin, P.~A. Golovach, D.~Sagunov, and K.~Simonov.
\newblock Algorithmic extensions of dirac's theorem.
\newblock In J.~S. Naor and N.~Buchbinder, editors, {\em Proceedings of the
  2022 {ACM-SIAM} Symposium on Discrete Algorithms, {SODA} 2022, Virtual
  Conference / Alexandria, VA, USA, January 9 - 12, 2022}, pages 406--416.
  {SIAM}, 2022.

\bibitem{jansen2019hamiltonicity}
B.~M. Jansen, L.~Kozma, and J.~Nederlof.
\newblock Hamiltonicity below dirac’s condition.
\newblock In {\em International Workshop on Graph-Theoretic Concepts in
  Computer Science}, pages 27--39. Springer, 2019.

\bibitem{knop2017simplified}
D.~Knop, M.~Kouteck{\`y}, T.~Masa{\v{r}}{\'\i}k, and T.~Toufar.
\newblock Simplified algorithmic metatheorems beyond {MSO}: treewidth and
  neighborhood diversity.
\newblock In {\em International Workshop on Graph-Theoretic Concepts in
  Computer Science}, pages 344--357. Springer, 2017.

\end{thebibliography}

%\newpage

%\section*{APPENDIX}

%\begin{figure}[!h]
%	\center
%	\includegraphics[width=.9\textwidth]{fig3}
%	\caption{Representation of a graph with $r-1$ cycles, a path of size $h$ and the edge $uw$ that will be added, creating a graph with $r$ cycles.}
%	\label{fig3}
%\end{figure}

%\vspace{-1.3cm}

%\begin{figure}[hb]
%\center
%\includegraphics[width=5cm]{case1}
%\caption{Representation of case 1.}
%\label{case1}
%\vspace{-3cm}
%\end{figure}

%\begin{figure}[!tb]
%    \center
%    \includegraphics[width=7cm]{case2}
%    \caption{Representation of case 2.}
%    \label{case2}
%\end{figure}

%\newpage

%\vspace{-.5cm}

\if 10

\begin{algorithm} [!htb]
\begin{scriptsize}
\caption{Single-exponential dynamic programming algorithm for {\sc Minimum Cycle Cover}.}
\label{cyclecoveralg}
%begin{algorithmic}[1]  
\SetKwInOut{Input}{input}
\SetKwInOut{Output}{output}

\Input{A graph $G$ having an isolated vertex $w$, and a table $M$ indexed by $X$, $s$, $t$ and $P2$ as previously described}
\Output{The minimum cycle cover number of $G$}
%\begin{algorithm}
%\caption{Dynamic programming for {\sc Cycle Cover}}\label{cyclecoveralg}
%\begin{algorithmic}

\Begin{

%\Procedure{CycleCover}{$G$, $w$}
Initialize $M$ setting $\infty$ on all its cells\;
\ForEach{$v \in V(G)$}{
    $M[\{v\}, v, v, 0] \gets 1$\;
}
\For{$x \gets 2\ to\ |V(G)|$}{
    \ForEach{ $X \subseteq V(G)$ such that $|X| = x$}{    
        \ForEach{$s \in X$}{
            \ForEach{$t \in X$}{
                \For{$P2 \gets 0\ to\ 1$}{
                    \If{($s = t$) and ($P2 = 1$)}{
                        $M[X, s, t, P2] \gets \infty$\;
                    }
                    \If{($s = t$) and ($P2 = 0$)}{
                         \ForEach{$s' \in X \setminus \{t\}$}{
                             \ForEach{$t' \in X \setminus \{t\}$}{
                                \If{($s't' \in E(G)$) or ($s'=t'$)}{
                                    $answer \gets M[X \setminus  \{t\}, s', t', 0] + 1$\;
                                    $M[X, s, t, P2] \gets \min\{M[X, s, t, P2], answer\}$\;
                                }
                            }
                        }
                    }
                    \If{($s \neq  t$) and ($P2 = 1$) and ($st \notin E(G)$)}{
                       $M[X, s, t, P2] \gets \infty$\;
                    }
                    \If{($s \neq  t$) and ($P2 = 1$) and ($st \in E(G)$)}{
                        $M[X, s, t, P2] \gets M[X \setminus  \{t\}, s, s, 0]$\;
                    }
                    \If{($s \neq  t$) and ($P2 = 0$)}{
                        \ForEach{$t' \in X \setminus \{t\}$ such that $tt' \in E(G)$}{
                            $answer \gets \min\{M[X \setminus  \{t\}, s, t', 0], M[X \setminus  \{t\}, s, t', 1]$\}\;
                            $M[X, s, t, P2] \gets \min\{M[X, s, t, P2], answer\}$\;
                        }
                    }
                }
            }
        }
    }        
}
\Return{$M[V(G), w, w, 0]$\;}

}
%\EndProcedure
%\end{algorithmic}
\end{scriptsize}
\end{algorithm}\vspace{-.5cm}

\begin{Claim}[proof in the appendix]\label{claim:correctness}
Algorithm~\ref{cyclecoveralg} is correct.
\end{Claim}

\noindent {\bf Proof of Claim~\ref{claim:correctness}.}
Let $\mathcal{C}=\{C_1,C_2,\ldots,C_r\}$ be a minimum cycle cover of a graph $G$ having an isolated vertex $w$. 
Let $C_j=v_{C_j}^1,v_{C_j}^2,\ldots,v_{C_j}^{|C_j|}$ for $1\leq j\leq r-1$, and let $C_r=w$.

\medskip

Let $X_i = V(C_1)\cup V(C_2)\cup \ldots \cup V(C_i)$, for $1\leq i\leq r$.

\smallskip

From the construction of the algorithm it follows that:

\begin{itemize}
    \item $M[X_1,v_{C_1}^1,v_{C_1}^{|C_1|},0]=1$; thus,
    \item $M[X_i,v_{C_i}^1,v_{C_i}^{|C_i|},0]\leq i$, for $1\leq i< r$;
    \item $M[X_k,w,w,0]\leq r$
\end{itemize}

Now, suppose that $M[X_k,w,w,0]<r$. Thus, from $M[X_k,w,w,0]$ there is a sequence $P$ of $n$ cells of the table $M$ that certifies it, where each cell have as predecessor in the sequence the cell that gave it the minimum. 

Let $c,c'$ be two consecutive cells of such a sequence, and let $s_c,t_c$, and $s_{c'},t_{c'}$ be the endpoints of the open path in $c$ and $c'$, respectively. 

If the value stored in $c$ is equal to the value stored in $c'$ then, by the recurrence relation, $t$ and $t'$ are adjacent in $G$. Also, if the value stored in $c$ is greater than the value stored in $c'$ then $t_{c'}$ and $s_{c'}$ are adjacent, implying that the cells $c^*$ in the sequence and having the same value stored in $c'$ certificate a cycle formed by each endpoint $t_{c^*}$ of each $c^*$. Hence, each $c$ with stored value greater than its predecessor $c'$ in $P$ certificates a cycle in $G$ form by $\{t_{c^*}: c^* \mbox{ has the same value stored in } c'\}$. Consequently, from $P$ one can build a cycle cover of $G$ with size less than $r$, contradicting that $\mathcal{C}$ is a minimum cycle cover.  Therefore, Algorithm~\ref{cyclecoveralg} computes the size of a minimum cycle cover. \qed

\fi

\end{document}